\documentclass[12pt,twoside]{article}

 \usepackage{float}
\usepackage{graphicx}
\usepackage{epstopdf}
\usepackage{graphicx}
\usepackage{epic}
\usepackage{multirow}
\usepackage{pst-poly}  
\usepackage{pst-plot}  
\usepackage{pst-poly}  
\usepackage{tikz}
\usepackage{threeparttable}
\usepackage{xcolor}
\usetikzlibrary{arrows,shapes,chains}

\renewcommand{\paragraph}{\roman{paragraph}}
\usepackage[a4paper]{geometry}
\makeatletter
\renewcommand\title[1]{\gdef\@title{\reset@font\Large\bfseries #1}}
\renewcommand\section{\@startsection {section}{1}{\z@}%
                                   {-3.5ex \@plus -1ex \@minus -.2ex}%
                                   {2.3ex \@plus.2ex}%
                                   {\normalfont\large\bfseries}}
\renewcommand\subsection{\@startsection{subsection}{2}{\z@}%
                                     {-3ex\@plus -1ex \@minus -.2ex}%
                                     {1.5ex \@plus .2ex}%
                                     {\normalfont\normalsize\bfseries}}
\renewcommand\subsubsection{\@startsection{subsubsection}{3}{\z@}%
                                     {-2.5ex\@plus -1ex \@minus -.2ex}%
                                     {1.5ex \@plus .2ex}%
                                     {\normalfont\normalsize\bfseries}}

\def\@runningauthor{}\newcommand{\runningauthor}[1]{\def\runningauthor{#1}}
\def\@runningtitle{}\newcommand{\runningtitle}[1]{\def\runningtitle{#1}}

\renewcommand{\ps@plain}{%
\renewcommand{\@evenhead}{\footnotesize\scshape \hfill\runningauthor\hfill}
\renewcommand{\@oddhead}{\footnotesize\scshape \hfill\runningtitle\hfill}}

\newcommand{\F}{\mathbb{F}}

\g@addto@macro\bfseries{\boldmath}

\makeatother

\usepackage{amsthm,amsmath,amssymb}
\usepackage{cite}
\usepackage{graphicx}

\usepackage[colorlinks=true,citecolor=black,linkcolor=black,urlcolor=blue]{hyperref}

\theoremstyle{plain}
\newtheorem{theorem}{Theorem}[section]

\newtheorem{lem}[theorem]{Lemma}
\newtheorem{cor}[theorem]{Corollary}
\newtheorem{prop}[theorem]{Proposition}

\theoremstyle{definition}

\newtheorem{example}[theorem]{Example}
\newtheorem{conjecture}[theorem]{Conjecture}

\theoremstyle{remark}
\newtheorem{remark}[theorem]{Remark}

\runningauthor{}

\date{}

\begin{document}

  \title{Several constructions of optimal LCD codes over small finite fields\thanks{This research is supported by Natural Science Foundation of China (12071001).}}
\author{ Shitao Li\thanks{lishitao0216@163.com, School of Mathematical Sciences, Anhui University, Hefei, 230601, China.}, Minjia Shi\thanks{smjwcl.good@163.com, Key Laboratory of Intelligent Computing and Signal Processing, Ministry of
Education, School of Mathematical Sciences, Anhui University, Hefei, 230601, China; State Key Laboratory of Information Security, Institute of Information Engineering, Chinese Academy of Sciences, Beijing,
100093, China.}, Huizhou Liu\thanks{18756027866@163.com, State Grid Anhui Electric Power Co., Ltd., Hefei, China.}}

\date{}
    \maketitle

\begin{abstract}
Linear complementary dual (LCD) codes are linear codes which intersect their dual codes trivially, which have been of interest and extensively studied due to their practical applications in computational complexity and information protection. In this paper, we give some methods for constructing LCD codes over small finite fields by modifying some typical methods for constructing linear codes. We show that all odd-like binary LCD codes, ternary LCD codes and quaternary Hermitian LCD codes can be constructed using the modified methods. Our results improve the known lower bounds on the largest minimum distances of LCD codes. Furthermore, we give two counterexamples to disprove the conjecture proposed by Bouyuklieva (Des. Codes Cryptogr. 89(11): 2445-2461, 2021).
\end{abstract}
\section{Introduction}
LCD codes were first introduced by Massey in 1992 to solve a problem in information theory \cite{LCD-7}.
In 2004, Sendrier \cite{LCD-is-good} showed that LCD codes meet the asymptotic Gilbert-Varshamov bound by using the hull dimension spectra of linear codes.
In 2016, Carlet and Guilley \cite{lcd-appl} investigated an application of binary LCD codes against Side-Channel Attacks (SCA) and Fault Injection Attack (FIA), and gave several constructions of LCD codes. They also have been employed to construct entanglement-assisted quantum error correction codes (see \cite{H-LCD-20,H-20}).
Recently, LCD codes were extensively studied \cite{LCD-new,LCD-2,LCD-3,
LCD-5,ACD-E-codes,LCD-T-matric}.
In particular, an interesting result is that Carlet {\it{et al.}} \cite{LCD-3} showed that any code over $\F_q$ is equivalent to some Euclidean LCD code for $q\geq 4$ and any code over $\F_{q^2}$ is equivalent to some Hermitian LCD code for $q\geq 3$. This also motivates us to study LCD codes, especially LCD codes over small finite fields.
Codes over $\F_2$, $\F_3$ and $\F_4$ are called binary codes, ternary codes and quaternary codes, respectively.

It is also a fundamental topic to determine the largest minimum distance of LCD codes for various lengths and dimensions in coding
theory.
In recent years, much work has been done concerning this fundamental topic.
The largest minimum distance among all binary LCD $[n,k]$ codes was partially determined
in \cite{bound-12,HS-BLCD-1-16,2-LCD-30,AH-BLCD-17-24,B-LCD-40,2-LCD-40} for $n\leq 40$. It is worth noting that there are many unknowns in \cite{B-LCD-40} and \cite{2-LCD-40}. The largest minimum distance among all ternary LCD $[n,k]$ codes was determined in \cite{AH-TLCD-1-10,ter-11-19} for $n\leq 20$. The largest minimum diatance among all quaternary Hermitian LCD $[n,k]$ codes was partially determined in \cite{H-LCD-20,H-20,H-25} for $n\leq 25$.
More constructions of LCD codes can be seen \cite{H-LCD-IT,H-lCD-k=3,AHS-BLCD,H-lcd-IJQI,Lin-sok,LCD-10,Koro,234-lcD}.
One of the interesting constructions is that Harada \cite{234-lcD} gave two methods for constructing many LCD codes from a given LCD code by modifying some known methods for constructing self-dual codes \cite{b-up-1,b-up-4,b-up-5,b-up-6}. Using these methods, many new binary LCD codes and quaternary Hermitian LCD were constructed.
Therefore, an open problem is to extend the above results and construct new LCD codes.

There are many typical methods for constructing new linear codes from old ones. For example, the puncturing and shortening techniques, subcodes construction.
In this paper, we give some methods for constructing LCD codes by modifying the above typical methods. We show that all odd-like binary LCD codes, ternary LCD codes and quaternary Hermitian LCD codes can be constructed by our construction methods. That is to say, our methods are efficient for constructing LCD codes.
Using theses methods,
we obtain that some binary LCD codes with better parameters comparing with \cite{B-LCD-40} and \cite{2-LCD-40}. In addition, we also obtain some binary LCD codes, which are not equivalent to the codes in \cite{2-LCD-40}.
We extend the tables on ternary LCD codes to lengths up to 25. Some ternary LCD codes with better parameters are constructed comparing with \cite{234-lcD}. Finally, we also construct some quaternary Hermitian LCD codes with new parameters comparing with \cite{H-25}. These codes improve the previously known lower bounds on the largest minimum weights.
It is worth mentioning that we give two counterexamples to disprove the conjecture proposed by Bouyuklieva \cite{B-LCD-40}.

The paper is organized as follows. In Section 2, we give some notations and definitions, which can be found in \cite{Huffman}. In Section 3, we give a method for constructing LCD $[n-\ell,k-\ell,\geq d]$ and $[n-\ell,k,d-\ell]$ codes from a given $[n,k,d]$ code with $\ell$-dimensional hull. We also give a general method for constructing LCD $[n+1,k]$ and $[n,k+1]$ codes from a given LCD $[n,k]$ code. In Sections 4-6, we construct binary LCD codes, ternary LCD codes, quaternary Hermitian LCD codes and the related entanglement-assisted quantum error correction codes.
In Section 7, we conclude the paper.
All computations in this paper have been done with the computer algebra
system MAGMA \cite{magma}.

\section{Preliminaries}

Let $\F_q$ denote the finite field with $q$ elements, where $q$ is a prime power.
For any $\textbf x\in \F_q^N$, the {\em Hamming weight} ${\rm wt}({\bf x})$ of $\textbf x$ is the
number of nonzero components of $\textbf x$.
An $[N,K,D]$ linear code $C$ over $\F_q$ is a $K$-dimensional subspace of $\F_q^N$, where $D$ is the minimum nonzero Hamming weight of $C$. Let $A_i(C)$ denote the number of codewords with Hamming weight $i$ in $C$, where $0\leq i\leq N$. Then the sequence $(A_0(C),A_1(C),\ldots,A_N(C))$ is called the {\em weight distribution} of $C$.
A {\em generator matrix} for an $[N,K]$ code $C$ is any $K\times N$ matrix $G$ whose rows form a basis for $C$. For any set of $k$ independent columns of a generator matrix $G$, the corresponding set of coordinates forms an {\em information set} for $C$.
A matrix is {\em monomial} if it contains exactly one nonzero element per row and per column.
Two codes $C$ and $D$ are {\em equivalent} if there is a monomial matrix $M$ such
that $MC = D$.
The {\em Euclidean dual} code $C^{\perp_E}$ of a linear code $C$ over $\F_q$ is defined as
$$C^{\perp_E}=\{\textbf y\in \F_q^N~|~\langle \textbf x, \textbf y\rangle_E=0, {\rm for\ all}\ \textbf x\in C \},$$
where $\langle \textbf x, \textbf y\rangle_E=\sum_{i=1}^N x_iy_i$ for $\textbf x = (x_1,x_2, \ldots, x_N)$ and $\textbf y = (y_1,y_2, \ldots, y_N)\in \F_q^N$.
The {\em Hermitian dual} code $C^{\perp_H}$ of a linear code $C$ over $\F_{q^2}$ is defined as
$$C^{\perp_H}=\{\textbf y\in \F_{q^2}^N~|~\langle \textbf x, \textbf y\rangle_H=0, {\rm for\ all}\ \textbf x\in C \},$$
where $\langle \textbf x, \textbf y\rangle_H=\sum_{i=1}^N x_i\overline{y_i}$ for $\textbf x = (x_1,x_2, \ldots, x_N)$ and $\textbf y = (y_1,y_2, \ldots, y_N)\in \F_{q^2}^N$. Note that $\overline{x}=x^q$ for any $x\in \F_{q^2}$.
The {\em Euclidean hull} (resp. {\em Hermitian hull}) of the linear code $C$ is defined as
$${\rm Hull_E}(C)=C\cap C^{\perp_E}\ ({\rm resp.\ Hull_H}(C)=C\cap C^{\perp_H}).$$

A linear code $C$ over $\F_q$ is called (Euclidean) LCD if $C\cap C^{\perp_E}=\{\textbf 0\}$.
A linear code $C$ over $\F_{q^2}$ is called Hermitian LCD if $C\cap C^{\perp_H}=\{\textbf 0\}$.
The following lemma is from \cite{LCD-6,LCD-7}.
\begin{lem}\label{lem-LCD}
{\rm(1)} Let $C$ be a code with the generator matrix $G$ over $\F_q$. Then $C$ is LCD
if and only if $GG^T$ is nonsingular, where $G^T$ denotes the transpose of $G$.\\
{\rm(2)} Let $C$ be a code with the generator matrix $G$ over $\F_{q^2}$. Then $C$ is Hermitian LCD if and only if $G\overline{G}^T$ is nonsingular, where $\overline{G}^T$ denotes the conjugate transpose of $G$.
\end{lem}

Throughout this paper, let $d^E_q (N,K)$ denote the largest minimum weight among all LCD
$[N, K]$ codes over $\F_q~(q = 2, 3)$, and let $d^H_4 (N, K)$ denote the largest minimum weight among all quaternary Hermitian LCD $[N,K]$ codes.
An LCD $[N, K]$ code over $\F_q$ is {\em optimal LCD} if it has the minimum weight $d^E_q (N,K)$. An LCD $[N,K]$ code over $\F_q$ is called {\em almost optimal LCD} if it has the minimum weight $d^E_q (N,K)-1$.

A vector $\textbf x = (x_1, x_2, \ldots , x_n)\in \F^n_2$
is {\em even-like} if $\sum_{i=1}^nx_i=0$ and is {\em odd-like} otherwise.
A binary code is said to be {\em even-like} if it
has only even-like codewords, and is said to be {\em odd-like} if it
is not even-like.

Let $C$ be an $[n, k, d]$ code over $\F_q$, and let $T$ be a set of $t$
coordinate positions in $C$. We puncture $C$ by deleting all the
coordinates in $T$ in each codeword of $C$. The resulting code
is still linear and has length $n - t$. We denote
the {\em punctured code} by $C^T$.
Consider the set
$C(T)$ of codewords which are $0$ on $T$; this set is a subcode of $C$. Puncturing $C(T)$ on $T$ gives a code over $\F_q$ of length $n -t$ called the {\em shortened code} on $T$ and denoted $C_T$.
The following lemma is also valid with respect to the Hermitian inner product.

\begin{lem}\label{lem-shorten-puncture}{\rm\cite{Huffman}}
Let $C$ be an $[n,k,d]$ code over $\F_q$. Let $T$ be a set of $t$ coordinates. Then:
\begin{itemize}
  \item [(i)] $(C^\perp)_T=(C^T)^\perp$ and $(C^\perp)^T=(C_T)^\perp$, and
  \item [(ii)] if $t<d$, then $C^T$ and $(C^\perp)_T$ have dimensions $k$ and $n- t- k$, respectively.
\end{itemize}
\end{lem}

\section{Several methods for constructing LCD codes}

\subsection{LCD codes from shortened codes and punctured codes of linear codes}
The puncturing and shortening techniques are two very important tools for constructing new codes from old ones. In this subsection, we will use these two techniques to construct
new LCD codes with interesting and new parameters from some old LCD codes.
Firstly, we prove that Lemma 22 in \cite{LCD-3} is valid with respect to the Hermitian inner product.

\begin{theorem}\label{thm-SO+LCD}
Any linear code $C$ over $\F_q$ (resp. $\F_{q^2}$) is the direct sum of a self-orthogonal code and an LCD code with respect to the Euclidean (resp. Hermitian) inner product.
\end{theorem}

\begin{proof}
We only consider the Hermitian inner product.
Let $\{\alpha_1,\ldots,\alpha_{\ell},\alpha_{\ell+1},\ldots,\alpha_k\}$ be a basis of $C$ such that
$\{\alpha_1,\ldots,\alpha_{\ell}\}$ is a basis of $C_1={\rm Hull_H}(C)=C\cap C^{\perp_H}$.
Let $C_2$ be a linear code generated by $\alpha_{\ell+1},\alpha_{l+2},\ldots,\alpha_{k}$.
Then $C=C_1\oplus C_2$.
For any ${\bf c}\in C_2\cap C_2^{\perp_H}$, we have $\langle {\bf c}, \alpha_i\rangle_H=0$. This implies that ${\bf c}\in  C\cap C^{\perp_H}$. Since $C$ is LCD, ${\bf c}={\bf0}$. Therefore, $C_2$ is LCD with respect to the Hermitian inner product.
This completes the proof.
\end{proof}

The following two theorems are very important and interesting.

\begin{theorem}\label{LCD-(n-l,k-l)}
If there exists an $[n,k,d]$ linear code $C$ with $\ell$-dimensional hull. Then there exists a set of $\ell$ coordinates position $T$ such that the shortened $C_T$ of $C$ on $T$ is an $[n-\ell,k-\ell,\geq d]$ LCD code with respect to the Euclidean and Hermitian inner product.
\end{theorem}

\begin{proof}
Let $G$ be a generator matrix of $C$. Without loss of generality, we may assume that $$G=(I_k|A)=({\bf e}_{k,i}|{\bf a}_i)_{1\leq i\leq k},$$
where ${\bf e}_{k,i}$ and ${\bf a}_i$ are the $i$-row of $I_k$ (the identity matrix) and $A$, respectively.
Assume that $\{{\bf r}_j\}_{j=1}^{\ell}$ is a basis of ${\rm Hull}(C)$ such that the first non-zero position of ${\bf r}_j$ is the $i_j$-th position. Without loss of generality, we may assume that
$1\leq i_1< i_2< \cdots < i_{\ell}$. Then $i_{\ell}\leq k$; otherwise ${\bf r}_{i_{\ell}}=\textbf 0$, which is a contradiction.

Let $T=\{i_1,i_2,\ldots,i_{\ell}\}$ and $J=\{1,2,\ldots,k\}\setminus T=\{j_1,j_2,\ldots,j_{k-\ell}\}$ such that $j_1<j_2<\cdots<j_{k-\ell}$.
Then we know that $\{{\bf r}_j\}_{j=1}^{\ell} \cup \{({\bf e}_{k,i}|{\bf a}_i)\}_{i\in J}$ is a basis of $C$.
So the code $C(T)$ by generating
$\{({\bf e}_{k,i}|{\bf a}_i)\}_{i\in J}$ is an LCD code by Theorem \ref{thm-SO+LCD}.
In fact, the generator matrix for the shortened code $C_T$ on $T$ is
$$G_T=({\bf e}_{k-l,i}|{\bf a}_{j_i})_{1\leq i\leq k-l},$$
 where ${\bf e}_{k-\ell,i}$ is the $i$-row of $I_{k-\ell}$ and ${\bf a}_{j_i}$ is the ${j_i}$-row of $A$ for $1\leq i\leq k-\ell$.
It follows from $C(T)$ is LCD that $C_T$ is LCD. The parameters for $C_T$ are obvious.
\end{proof}

\begin{theorem}\label{puncture}
Let $C$ be an $[n,k,d]$ linear code with $\ell$-dimensional hull. If $t<d$, then there exists a set of $t$ coordinate position $T$ such that the punctured code $C^T$ of $C$ on $T$ is an LCD $[n-\ell,k,d^*\geq d-\ell]$ code with respect to the Euclidean and Hermitian inner product.
\end{theorem}

\begin{proof}
The parameters for the punctured code $C^T$ of $C$ on $T$ are obvious from (ii) of Lemma \ref{lem-shorten-puncture}.
Since $C$ is an $[n,k]$ linear code with $\ell$-dimensional hull, $C^{\perp}$ is an $[n,n-k]$ linear code with $\ell$-dimensional hull.
By Theorem \ref{LCD-(n-l,k-l)}, there exists a set of $\ell$ coordinate positions $T$ such that the shortened code $(C^\perp)_T$ of $C^\perp$ on $T$ is an LCD $[n-\ell,n-k-\ell]$ code.
It follows from (i) of Lemma \ref{lem-shorten-puncture} that $(C^T)^\perp=(C^\perp)_T$ is an LCD $[n-\ell,n-k-\ell]$ code.
It turns out that $C^T=((C^T)^\perp)^\perp$ is an LCD $[n-\ell,k]$ code.
\end{proof}

\begin{remark}
By Theorem \ref{LCD-(n-l,k-l)}, $C_T$ is LCD if $T$ is an information set of ${\rm Hull}(C)$.
Compared to the randomness of \cite{2-LCD-40} and \cite{H-25}, we determine the shortened set $T$.
In addition, the corollary 25 in \cite{LCD-3} proved that there exists an $[n+\ell,k,\geq d]$ Euclidean LCD code if there exists an $[n,k,d]$ code with $\ell$-dimensional Euclidean hull. Hence, it is easy to see that Theorem \ref{LCD-(n-l,k-l)} is more effective than \cite[Corollary 25]{LCD-3}.
\end{remark}

\subsection{Construction method I}
Firstly, we recall a typical construction method of linear codes. Let $C$ be a linear $[n,k]$ code with the generator matrix $G$. Then for any ${\bf x}\in \F_q^n$, the following matrix
$$
G'=\begin{pmatrix}
1 & {\bf x} \\
 {\rm{\bf 0}} & G
\end{pmatrix}
$$
generates a linear $[n+1,k+1]$ code. In addition, any linear $[n+1,k+1]$ code over $\F_q$ can be obtained from some linear $[n,k]$ code using the above construction.

In this subsection, we give a method for constructing many LCD $[n+1,k+1]$ codes with interesting parameters from a given LCD $[n,k]$ code by modifying the above construction method of linear codes.

\begin{theorem}\label{Methods}
{\rm(1)} Let $C$ be a binary LCD $[n,k]$ code with the generator matrix $G$. Let ${\bf x}\in C^\perp$ such that ${\rm wt}({\bf x})$ is even. Then the following matrix
$$
G'=\begin{pmatrix}
1 & {\bf x} \\
 {\rm{\bf 0}} & G
\end{pmatrix}
$$
generates a binary odd-like LCD $[n+1,k+1]$ code.

{\rm(2)} Let $C$ be a ternary LCD $[n,k]$ code with the generator matrix $G$. Let ${\bf x}\in C^\perp$ such that ${\rm wt}({\bf x})\neq 2\ ({\rm mod}\ 3)$. Then the following matrix
$$
G'=\begin{pmatrix}
1 & {\bf x} \\
 {\rm{\bf 0}} & G
\end{pmatrix}
$$
generates a ternary LCD $[n+1,k+1]$ code.

{\rm(3)} Let $C$ be a quaternary Hermitian LCD $[n,k]$ code with the generator matrix $G$. Let ${\bf x}\in C^{\perp_H}$ such that ${\rm wt}({\bf x})$ is even. Then the following matrix
$$
G'=\begin{pmatrix}
1 & {\bf x} \\
 {\rm{\bf 0}} & G
\end{pmatrix}.
$$
generates a quaternary Hermitian LCD $[n+1,k+1]$ code.
\end{theorem}

\begin{proof}
We only give the proof of (1), the other cases are similar.
Let ${\bf r}_i$ be the $i$-th row of $G$ for $1\leq i\leq k$. Let ${\bf r}'_j$ be $j$-th row of $G'$ for $1\leq j\leq k+1$.
Then we have
\begin{align*}
  \langle {\bf r}'_1, {\bf r}'_1\rangle_E  =&1+\langle {\bf x}, {\bf x}\rangle_E=1, \\
  \langle {\bf r}'_1, {\bf r}'_j\rangle_E  =&\langle {\bf x}, {\bf r}_{j-1}\rangle_E=0\ {\rm for}\ 2\leq j\leq k+1,\\
  \langle {\bf r}'_j, {\bf r}'_j\rangle_E  =&\langle {\bf r}_{j-1}, {\bf r}_{j-1}\rangle_E\ {\rm for}\ 2\leq j\leq k+1,\\
  \langle {\bf r}'_j, {\bf r}'_{j'}\rangle_E  =&\langle {\bf r}_{j-1}, {\bf r}_{j'-1}\rangle_E\ {\rm for}\ 2\leq j<j'\leq k+1.
\end{align*}
Hence we obtain that
$$
G'G'^{T}=\begin{pmatrix}
1 & 0\ \cdots\ 0 \\
\begin{array}{c}
  0 \\
  \vdots \\
  0
\end{array}
 & GG^T
\end{pmatrix}.
$$
If $C$ is a binary LCD code, then it follows from Lemma \ref{lem-LCD} that $GG^{T}$ is nonsingular. This implies that
$G'G'^{T}$ is nonsingular. By Lemma \ref{lem-LCD}, $G'$ generates a binary LCD $[n+1,k+1]$ code.
This completes the proof.
\end{proof}

It is easy to see that all binary LCD codes constructed by Theorem \ref{Methods} are odd-like.
Next, we prove that all odd-like binary LCD codes can be obtained by
the construction in Theorem \ref{Methods}.

\begin{theorem}
Any binary odd-like LCD $[n,k]$ code is obtained from some binary LCD $[n-1,k-1]$ code (up to equivalence) using the construction of Theorem \ref{Methods}.
\end{theorem}

\begin{proof}
Let $C$ be a binary odd-like LCD $[n,k]$ code. According to \cite[Propsition 3]{B-LCD-40}, there is at least one coordinate position $i$ such that the shortened code $C_{\{i\}}$ of $C$ on the $i$-th coordinate is a binary LCD $[n-1,k-1]$ code. Without loss of generality, we consider that $i=1$.
Assume that $C_{\{1\}}$ has the generator matrix $G_1$. Then
there exists a binary vector ${\bf x'}=(x'_1,x'_2,\ldots,x'_{n-1})$ of length $n-1$ such that the following matrix
$$G=\begin{pmatrix}
1 & {\bf x'} \\
\textbf 0 & G_1
\end{pmatrix}$$
is a generator matrix of $C$.
Since $C_{\{1\}}$ is a binary LCD code, $\F_2^{n-1}=C_{\{1\}}\oplus {C_{\{1\}}}^\perp$.
So there are ${\bf x}=(x_1,x_2,\ldots,x_{n-1})\in {C_{\{1\}}}^\perp$ and ${\bf y}=(y_1,y_2,\ldots,y_{n-1})\in C_{\{1\}}$ such that ${\bf x'}={\bf x}+{\bf y}$.

Using the binary vector ${\bf x}=(x_1,x_2,\ldots,x_{n-1})\in {C_{\{1\}}}^\perp$ of length $n-1$ and $G_1$ we get a generator matrix $G'$ of a
linear $[n, k]$ code $C'$ by Theorem \ref{Methods}. And
$$G'=\begin{pmatrix}
1 & {\bf x} \\
\textbf 0 & G_1
\end{pmatrix}\sim
\begin{pmatrix}
1 & {\bf x}+{\bf y} \\
\textbf 0 & G_1
\end{pmatrix}=G.$$
Thus the given code $C$ is equivalent to $C'$, as desired.
We assert that $wt(x)$ is even; otherwise $(1|{\bf x})\in C\cap {C}^\perp$, which is a contradiction.
This completes the proof.
\end{proof}

Finally, we give two counterexamples to show that Conjecture 1 in \cite{B-LCD-40} is invalid.

\begin{conjecture}\label{conjecture}{\rm\cite[Conjecture 1]{B-LCD-40}}
Let $k$ be an even position integer and $n>k$ be another integer. If $d_{2}^{E}(n,k)$ is even and $d_2^{E}(n-1,k)=d_2^{E}(n,k)-1$, then all binary LCD $[n,k,d_2^{E}(n,k)]$ are even-like codes.
\end{conjecture}

\begin{prop}\label{anti-example}
There exist optimal binary odd-like LCD $[14,8,4]$ and $[16,10,4]$ codes.
\end{prop}

\begin{proof}
According to \cite[Table 3]{HS-BLCD-1-16}, there exist optimal binary odd-like LCD $[13,7,4]$ and $[15,9,4]$ codes, which have the following generator matrices respectively
{\small\begin{center}
$\begin{pmatrix}
\setlength{\arraycolsep}{1.2pt}
    \begin{array}{ccccccccccccc}
      1&0&0&0&0&0&0&1&1&0&1&0&1 \\
        0&1&0&0&0&0&0&0&1&1&1&1&0 \\
        0&0&1&0&0&0&0&0&0&1&1&0&1 \\
        0&0&0&1&0&0&0&1&0&0&1&1&0 \\
        0&0&0&0&1&0&0&0&0&1&0&1&1 \\
        0&0&0&0&0&1&0&1&1&1&0&1&1 \\
        0&0&0&0&0&0&1&1&1&1&1&0&0 \\
    \end{array}
\end{pmatrix},~
$
$\begin{pmatrix}
\setlength{\arraycolsep}{1.2pt}
   \begin{array}{ccccccccccccccc}
    1&0&0&0&0&0&0&0&0&1&1&0&0&0&1\\
0&1&0&0&0&0&0&0&0&0&1&1&1&1&0\\
0&0&1&0&0&0&0&0&0&0&1&1&1&0&1\\
0&0&0&1&0&0&0&0&0&1&1&0&0&1&0\\
0&0&0&0&1&0&0&0&0&0&1&1&0&1&1\\
0&0&0&0&0&1&0&0&0&0&1&0&1&1&1\\
0&0&0&0&0&0&1&0&0&1&1&1&1&1&1\\
0&0&0&0&0&0&0&1&0&1&1&0&1&0&0\\
0&0&0&0&0&0&0&0&1&1&1&1&0&0&0\\
   \end{array}
\end{pmatrix}.
$
\end{center}}
By applying Theorem \ref{Methods} to the LCD $[13,7,4]$ code, a binary LCD $[14,8,4]$ code $C'$ is constructed, where ${\bf x}=(1 0 0 1 1 1 0 0 0 1 1 0 0)$. The code $C'$ has weight distribution as:
$$[ \langle0, 1\rangle, \langle4, 24\rangle, \langle5, 36\rangle, \langle6, 36\rangle, \langle7, 60\rangle, \langle8, 45\rangle, \langle9, 28\rangle, \langle10, 20\rangle, \langle11,4\rangle, \langle12, 2\rangle].$$
Similarly, a binary LCD $[16,10,4]$ code is constructed, where ${\bf x}=(1 1 1 1 1 1 0 1 1 0 0 1 1 1 1)$, which has weight distribution as:
\begin{center}
$[ \langle0, 1\rangle, \langle4, 43\rangle, \langle5, 81\rangle, \langle6, 96\rangle, \langle7, 189\rangle, \langle8, 207\rangle, \langle9, 162\rangle,$\\ $\langle10, 144\rangle,
\langle11, 66\rangle, \langle12, 21\rangle, \langle13, 13\rangle, \langle15, 1\rangle ].$
\end{center}
In the sense of equivalence, these two LCD code are also constructed by \cite{HS-BLCD-1-16}. Here we just show that we can also get such codes by our method.
\end{proof}

\begin{remark}
Let $n=14$ and $k=8$.
By \cite[Table 3]{HS-BLCD-1-16}, $d_{2}^E(14,8)=4$ and $d_{2}^E(13,8)=3$. The condition of Conjecture \ref{conjecture} is satisfied. However, the binary LCD $[14,8,4]$ code in Proposition \ref{anti-example} is odd-like. Therefore, Conjecture {\rm\ref{conjecture}} is invalid.
Similarly, the odd-like binary LCD $[16,10,4]$ is also a counterexample of Conjecture \ref{conjecture}.
\end{remark}

\subsection{Construction method II}
In this subsection, we recall another typical construction methods of linear codes. Let $C$ be a linear $[n,k]$ code with the generator matrix $G$ over $\F_q$. For any ${\bf x}\in \F_q^n\setminus C$, the following matrix
$$
G'=\begin{pmatrix}
 {\bf x} \\
  G
\end{pmatrix}
$$
generates a linear $[n,k+1]$ code. In addition, any linear $[n,k+1]$ code over $\F_q$ can be obtained from some linear $[n,k]$ code using the above construction.

Next, we give a method for constructing many LCD $[n,k+1]$ codes with interesting parameters from a given LCD $[n,k]$ code by modifying the above construction method of linear codes.
\begin{theorem}\label{Methods-2}
{\rm(1)} Let $C$ be a binary LCD $[n,k]$ code with the generator matrix $G$. Let ${\bf y}\in C^\perp$ such that ${\rm wt}({\bf y})$ is odd. Then the following matrix
$$
G'=\begin{pmatrix}
 {\bf y} \\
  G
\end{pmatrix}.
$$
generates a binary LCD $[n,k+1]$ code.

{\rm(2)} Let $C$ be a ternary LCD $[n,k]$ code with the generator matrix $G$. Let ${\bf y}\in C^\perp$ such that ${\rm wt}({\bf y})\neq 0\ ({\rm mod}\ 3)$. Then the following matrix
$$
G'=\begin{pmatrix}
 {\bf y} \\
  G
\end{pmatrix}.
$$
generates a ternary LCD $[n,k+1]$ code.

{\rm(3)} Let $C$ be a quaternary Hermitian LCD $[n,k]$ code with the generator matrix $G$. Let ${\bf y}\in C^{\perp_H}$ such that ${\rm wt}({\bf y})$ is odd. Then the following matrix
$$
G'=\begin{pmatrix}
 {\bf y} \\
 G
\end{pmatrix}.
$$
generates a quaternary Hermitian LCD $[n,k+1]$ code.
\end{theorem}

\begin{proof}
The proof is similar to that of Theorem \ref{Methods}, so we omit it here.
\end{proof}

Next, we consider the converse of Theorem \ref{Methods-2}.

\begin{theorem}\label{up-2}
Any binary odd-like LCD $[n,k,d]$ code is obtained from some binary LCD $[n,k-1,\geq d]$ code (up to equivalence) by the construction of Theorem \ref{Methods-2}.
\end{theorem}

\begin{proof}
Let $C$ be a binary odd-like LCD $[n,k,d]$ code with the generator matrix $G$. From \cite[Theorem 3]{LCD-new}, there exists a basis ${\bf c}_1,{\bf c}_2,\ldots,{\bf c}_{k}$ of $C$ such that ${\bf c}_i \cdot {\bf c}_j=1$ if $i=j$ and ${\bf c}_i\cdot {\bf c}_j=0$ otherwise. Let $C_0$ be a binary code with the generator matrix $G_0=({\bf c}_{i})_{1\leq i\leq k-1}$, where $c_{i}$ is the $i$-th row of $G_0$. From \cite[Theorem 3]{LCD-new}, $C_0$ is LCD. This implies that $\F_{2}^n=C_0 \oplus C_0^\perp$.
So there are ${\bf x}=(x_1,x_2,\ldots,x_n)\in C_0$ and ${\bf y}=(y_1,y_2,\ldots,y_n)\in C_0^\perp$ such that ${\bf c}_k={\bf x}+{\bf y}$. Since ${\bf c}_k\notin C_0$, ${\bf y}={\bf c}_k-{\bf x}\notin C_0$.

Using the binary vector ${\bf y}=(y_1,y_2,\ldots,y_{n})\in {C_0}^\perp$ of length $n$ and $G_0$ we get a generator matrix $G'$ of a
linear $[n, k+1]$ code $C'$ by Theorem \ref{Methods-2}. And
$$G'=\begin{pmatrix}
 {\bf y} \\
 G_0
\end{pmatrix}\sim
\begin{pmatrix}
{\bf x}+{\bf y} \\
 G_0
\end{pmatrix}=\begin{pmatrix}
{\bf c}_k \\
 G_0
\end{pmatrix}\sim G.$$
Thus the given code $C$ is equivalent to $C'$, as desired.
It turns out that $wt({\bf y})$ is odd; otherwise ${\bf y}\in C_0\cap C_0^{\perp}$, which is a contradiction.
This completes the proof.
\end{proof}

\begin{theorem}
Any ternary Euclidean (resp. quaternary Hermitian) LCD $[n,k]$ code is obtained from some ternary Euclidean (resp. quaternary Hermitian) LCD $[n,k-1]$ code by the construction of Theorem \ref{Methods-2}.
\end{theorem}

\begin{proof}
According to \cite[Propsition 4]{34-LCD}, we know that any ternary LCD $[n,k]$ code contains a ternary LCD $[n,k-1]$ subcode and any quaternary Hermitian LCD $[n,k]$ code contains a quaternary Hermitian LCD $[n,k-1]$ subcode.
The rest of the proof is similar to that of Theorem \ref{up-2}, so we omit it here.
\end{proof}

\section{New binary LCD codes}

The constructions of optimal binary LCD $[n,k]$ codes were studies
in \cite{bound-12,HS-BLCD-1-16,2-LCD-30,AH-BLCD-17-24,B-LCD-40,2-LCD-40,234-lcD}. It is worth noting that there are many unknowns in \cite{B-LCD-40} and \cite{2-LCD-40}. Hence, we will construct these unknown binary LCD codes by Constructions I and II in Section 3.

\subsection{Some important inequalities}
Let $d^E_2 (n, k)$ denote the largest minimum weight among all binary Euclidean LCD $[n, k]$ codes.
By adding zero column, the following inequality is obvious.

\begin{lem}\label{lem-inequality-1}
 Suppose that $ k\leq n$. Then $d^E_2 (n+1, k)\geq d^E_2 (n, k)$.
\end{lem}

The following lemmas are nontrivial and they were proved in \cite{LCD-new} and \cite{B-LCD-40}.
\begin{lem}{\rm\cite[Theorem 8]{LCD-new}}\label{lemma-leq}
Suppose that $2\leq k\leq n$. Then
$d^E_2(n,k)\leq d^E_2(n,k-1).$
\end{lem}

\begin{lem}{\rm\cite{B-LCD-40}}\label{lemma-leq-2}
Let $k$ and $n$ be two integers such that $1\leq k\leq n$. Then we have
\begin{itemize}
  \item [{\rm(1)}] If $k$ is odd, then $d^E_2(n,k)\leq d^E_2(n-1,k-1)$.
  \item [{\rm(2)}] If $k$ is even and $d^E_2(n,k)$ is odd, then $d^E_2(n+1,k)\geq d^E_2(n,k)+1$.
  \item [{\rm(3)}] If $d^E_2(n,k)$ is odd, then $d^E_2(n+2,k)\geq d^E_2(n,k)+1$.
\end{itemize}
\end{lem}

Combining Theorem \ref{LCD-(n-l,k-l)} and \cite[Proposition 5]{B-LCD-40}, we give an improved upper bound.
\begin{cor}\label{cor-4.4}
Let $k$ and $n$ be two integers such that $2\leq k\leq n$. Then we have
$$d^E_2(n,k)\leq \max\{d^E_2(n-1,k-1),d^E_2(n-2,k-2)\}.$$
\end{cor}
\begin{proof}
Assume that $C$ is a binary $[n,k,d]$ LCD code. Let $C_i$ be the
shortened code of $C$ on $i$-th coordinate for some $i\in \{1,2,\ldots,n\}$. By \cite[Proposition 5]{B-LCD-40}, we know that $\dim({\rm Hull}(C_i))\leq 1$.
If $\dim({\rm Hull}(C_i))=0$, then there exists a binary $[n-1,k-1,d]$ LCD code.
If $\dim({\rm Hull}(C_i))=1$, then there exists a binary $[n-1,k-1,d]$ code with one-dimension hull. By Theorem \ref{LCD-(n-l,k-l)}, there exists a binary $[n-2,k-2,d]$ LCD code.
 This completes the proof.
\end{proof}

Let $C^{\ell}_{[n,k,d]_2}$ denote a binary $[n,k,d]$ linear code with the generator matrix $G^{\ell}_{[n,k,d]_2}$, where $\ell$ is the dimension of hull for such code. When $\ell=0$, let $C_{[n,k,d]_2}=C^{0}_{[n,k,d]_2}$ denote a binary LCD $[n,k,d]$ code.

\subsection{New binary LCD codes}

From \cite{B-LCD-40}, only the exact value of $d_{2}^E(29,11)$ remains unknown for $n=29$.

\begin{prop}\label{29}
There is a binary LCD $[29,11,9]$ code.
\end{prop}
\begin{proof}
By the MAGMA function BKLC, there is a binary linear $[30,11,10]$ code $C^{1}_{[30,11,10]_2}$ with one-dimensional hull. By Theorem \ref{LCD-(n-l,k-l)}, the shortened code $(C^{1}_{[30,11,10]_2})_{\{1\}}$ on the first coordinate position is a binary LCD $[29,10,10]$ code $C_{[29,10,10]_2}$.
We apply Construction II to the code $C_{[29,10,10]_2}$. Let us take ${\bf y}=(0 0 0 1 0 1 0 1 1 0 1 1 1 0 1 0 1 0 0 1 1 0 0 0 1 0 0 0 1)$, one can construct a binary LCD $[29,11,9]$ code $C_{[29,11,9]_2}$.
\end{proof}

 Although the binary LCD code $C_{[29,11,9]_2}$ in Proposition \ref{29} has the same parameters with the binary LCD $[29,11,9]$ code in \cite[Proposition 3]{2-LCD-40}, their construction methods are different. In addition, the parameters in bold in Tables 1-3 denote that the corresponding code has new parameters according to \cite{B-LCD-40} and \cite{2-LCD-40}.

\begin{prop}\label{[30,12,9]-[30,15,7]}
There are binary LCD $[30,11,9]$ and $[30,15,7]$ codes.
\end{prop}

\begin{proof}
It follows from Lemma \ref{lem-inequality-1} and Proposition \ref{29} that there is a binary LCD $[30,11,9]$ code.
By the MAGMA function BKLC, there is a binary linear $[36,21,7]$ code $C_{[36,21,7]_2}^6$ with 6-dimensional hull.
By Theorem \ref{LCD-(n-l,k-l)}, the shortened code $(C_{[36,21,7]_2}^6)_T$ of $C_{[36,21,7]_2}^6$ is a binary LCD $[30,15,7]$ code $C_{[30,15,7]_2}$, where $T=\{1,2,3,5,17,19\}$.
\end{proof}

\begin{remark}
We list in Tables 1-2 some binary LCD codes obtained by Theorems \ref{LCD-(n-l,k-l)}, \ref{Methods} and \ref{Methods-2}. To save the space, the codes in Tables 1 and 2 can be obtained from {\tt https://ahu-coding.github.io/code1/}.
\end{remark}

{\small
\begin{center}
\begin{threeparttable}
\begin{tabular}{llll}
\multicolumn{4}{c}{{\rm Table 1: The binary LCD codes from Theorem \ref{LCD-(n-l,k-l)}}}\\
\hline
    \makebox[0.25\textwidth][l]{{\rm The\ code}\ $C_{[n,k,d]_2}^{\ell}$}& \makebox[0.3\textwidth][l]{The set $T$} &
    \makebox[0.18\textwidth][l]{$C_T$} &
    \makebox[0.15\textwidth][l]{References}\\
    \hline\hline
$C^{1}_{[30,11,10]_2} $&$ \{1\} $&$ [29,10,10]$&Theorem \ref{LCD-(n-l,k-l)}\\

$C_{[36,21,7]_2}^6 $&$ \{1,2,3,5,17,19\} $&$ {\bf[30,15,7]}$&Theorem \ref{LCD-(n-l,k-l)}\\

$C^{1}_{[32,21,6]_2} $&$ \{1\} $&$ [31,20,6]$&Theorem \ref{LCD-(n-l,k-l)}\\

$C^{1}_{[34,15,9]_2} $&$ \{1\} $&$ [33,14,9]$&Theorem \ref{LCD-(n-l,k-l)}\\

$C^{0}_{[33,22,6]_2} $&$\backslash $&$ [33,22,6]$&Theorem \ref{LCD-(n-l,k-l)}\\

$C^{0}_{[34,22,6]_2} $&$ \backslash $&$ [34,22,6]$&Theorem \ref{LCD-(n-l,k-l)}\\

$C^{0}_{[36,16,10]_2} $&$ \backslash $&$ [36,16,10]$&Theorem \ref{LCD-(n-l,k-l)}\\

$C^{7}_{[43,27,7]_2} $&$ \{1,2,3,5,9,17,25\} $&$ {\bf[36,20,7]}$&Theorem \ref{LCD-(n-l,k-l)}\\

$C^{3}_{[40,13,13]_2} $&$ \{2,3,5\} $&$ {\bf[37,10,13]}$&Theorem \ref{LCD-(n-l,k-l)}\\

$C^{8}_{[45,29,7]_2} $&$ \{1,2,3,5,9,17,21,25\} $&$ {\bf[37,21,7]}$&Theorem \ref{LCD-(n-l,k-l)}\\
\hline
\end{tabular}
\begin{tablenotes}
\footnotesize
\item Note that the code $C_{[n,k,d]_2}^{\ell}$ is the best-known binary $[n,k,d]$ code from MAGMA \cite{magma}.
\end{tablenotes}
\end{threeparttable}
\end{center}

}

{\small
\begin{center}
\begin{threeparttable}
\begin{tabular}{llll}
\multicolumn{4}{c}{{\rm Table 2: The binary LCD codes from Theorem \ref{Methods} and Theorem \ref{Methods-2}}}\\
\hline
    \makebox[0.17\textwidth][l]{{\rm The\ given}\ $C_{[n,k,d]_2}$}& \makebox[0.22\textwidth][l]{The vector ${\bf x}$ or ${\bf y}$} &
    \makebox[0.1\textwidth][l]{$C'$} &
    \makebox[0.1\textwidth][l]{References}\\
    \hline\hline
[30,15,7] & (110001101111110101111111001001) & {\bf[31,16,7]} & {\rm Theorem}\ \ref{Methods} \\

[31,16,7] & (1110101000111010001110111011001) & ${\bf[32,17,7]}$& {\rm Theorem}\ \ref{Methods} \\

[32,17,7] & (11010011101011101011011000110111) & {\bf[33,18,7]}& {\rm Theorem}\ \ref{Methods} \\

[33,14,9] & (111000000011111011010110010000110) & [34,15,9]& {\rm Theorem}\ \ref{Methods} \\

[34,18,8] & (1100000000000110001000101101010110) & {\bf[35,19,7]}&{\rm Theorem}\ \ref{Methods} \\

[34,22,6] &  (1111111111111111111111111111111111) & {\bf[35,23,6]}& {\rm Theorem}\ \ref{Methods} \\

[36,16,10] &  (110100000100001000011010100001101110) & [37,17,9]&{\rm Theorem}\ \ref{Methods} \\

[37,21,7] & (0011101010001001101010110111111001111) & {\bf[38,22,7]}&{\rm Theorem}\ \ref{Methods} \\

[38,10,14] &  (11110000110000000000000000000111111110) & [39,11,13]&{\rm Theorem}\ \ref{Methods} \\

[38,22,7] & (10011010101101001101010111000000101111) & {\bf[39,23,7]}&{\rm Theorem}\ \ref{Methods} \\

[29,10,10] & (00010101101110101001100010001) & [29,11,9] &{\rm Theorem}\ \ref{Methods-2} \\

[31,20,6] & (0000111101010000100101010011100) & [31,21,5] &{\rm Theorem}\ \ref{Methods-2} \\
\hline
\end{tabular}
\begin{tablenotes}
\footnotesize
\item Note that the given code $C_{[n,k,d]_2}$ in Table 2 is from Tables 1-2. In addition, the codes $C_{[34,18,8]_2}$ and $C_{[38,10,14]_2}$ are the extendcodes of the codes $C_{[33,18,7]_2}$ and $C_{[37,10,13]_2}$, respectively (see Table 3).
\end{tablenotes}
\end{threeparttable}
\end{center}}

We list in Table 3 some binary LCD codes by using some inequalities, where the given code $C$ in Table 3 is from Tables 1-3.

{\small
\begin{center}
\begin{tabular}{lll||lll}
\multicolumn{6}{c}{{\rm Table 3: The binary LCD codes from some inequalities}}\\
\hline
    \makebox[0.12\textwidth][l]{{\rm The\ given}\ $C$}&
    \makebox[0.12\textwidth][l]{$C'$} &
    \makebox[0.1\textwidth][l]{References} &
   \makebox[0.12\textwidth][l]{{\rm The\ given}\ $C$}&
    \makebox[0.12\textwidth][l]{$C'$} &
    \makebox[0.1\textwidth][l]{References} \\
    \hline\hline
[29,11,9] & [30,11,9] & {\rm Lemma\ \ref{lem-inequality-1}}&
[35,21,6] & [36,21,6] & {\rm Lemma\ \ref{lem-inequality-1}}\\

[29,11,9] & [31,11,10] & {\rm (3)\ of\ Lemma\ \ref{lemma-leq-2}}&[35,23,6] & [36,23,6] & {\rm Lemma\ \ref{lem-inequality-1}}\\

[31,16,7] & [31,15,7] & {\rm Lemma\ \ref{lemma-leq}}&[36,18,8]&[37,18,8]&{\rm Lemma\ \ref{lem-inequality-1}}\\

[31,16,7] & {\bf[32,16,8]} & {\rm (2)\ of\ Lemma\ \ref{lemma-leq-2}}&
[35,19,7]&{\bf[37,19,8]}&{\rm (3)\ of\ Lemma\ \ref{lemma-leq-2}}\\

[32,16,8] & {\bf[32,15,8]} & {\rm Lemma\ \ref{lemma-leq}}&
[36,20,7]&{\bf[37,20,8]}& {\rm (2)\ of\ Lemma\ \ref{lemma-leq-2}}\\

[31,21,5] & [32,21,5] & {\rm Lemma\ \ref{lem-inequality-1}}&
[36,23,6]&[37,23,6]&{\rm Lemma\ \ref{lem-inequality-1}}\\

[32,15,8] & [33,15,8] & {\rm Lemma\ \ref{lem-inequality-1}}&
[37,10,13]&{\bf[38,10,14]}&{\rm (2)\ of\ Lemma\ \ref{lemma-leq-2}}\\

[32,16,8] & {\bf[33,16,8]} & {\rm Lemma\ \ref{lem-inequality-1}}&
[37,17,9]&[38,17,9]&{\rm Lemma\ \ref{lem-inequality-1}}\\

[32,17,7] & {\bf[33,17,7]} & {\rm Lemma\ \ref{lem-inequality-1}}&
[37,19,8]&[38,19,8]&{\rm Lemma\ \ref{lem-inequality-1}}\\

[33,22,6] & {\bf[33,21,6]} & {\rm Lemma\ \ref{lemma-leq}}&[37,20,8]&{\bf[38,20,8]}&{\rm Lemma\ \ref{lem-inequality-1}}\\

[33,18,7] & {\bf[34,18,8]} & {\rm (2)\ of\ Lemma\ \ref{lemma-leq-2}}&
[37,21,7]&[38,21,7]&{\rm Lemma\ \ref{lem-inequality-1}}\\

[34,18,8] & {\bf[34,17,8]} & {\rm Lemma\ \ref{lemma-leq}}&[37,23,6]&[38,23,6]&{\rm Lemma\ \ref{lem-inequality-1}}\\

[34,17,8] & [34,16,8] & {\rm Lemma\ \ref{lemma-leq}}&
[38,10,14]&[39,10,14]&{\rm Lemma\ \ref{lem-inequality-1}}\\

[33,21,6] & [34,21,6] & {\rm Lemma\ \ref{lem-inequality-1}}&
[38,20,8]&[39,20,8]&{\rm Lemma\ \ref{lem-inequality-1}}\\

[34,17,8] & {\bf[35,17,8]} & {\rm Lemma\ \ref{lem-inequality-1}}&
[38,22,7]&{\bf[39,22,8]}&{\rm (2)\ of\ Lemma\ \ref{lemma-leq-2}}\\

[34,18,8] & {\bf[35,18,8]} & {\rm Lemma\ \ref{lem-inequality-1}}&[39,22,8]&{\bf[39,21,8]}&{\rm Lemma\ \ref{lemma-leq}}\\

[34,21,6] & {\bf[35,21,6]} & {\rm Lemma\ \ref{lem-inequality-1}}&
[39,10,14]&[40,10,14]&{\rm Lemma\ \ref{lem-inequality-1}}\\

[34,15,9] & [36,15,10] & {\rm (3)\ of\ Lemma\ \ref{lemma-leq-2}}&[39,11,13]&[40,11,13]&{\rm Lemma\ \ref{lem-inequality-1}}\\

[35,17,8] & [36,17,8] & {\rm Lemma\ \ref{lem-inequality-1}}&
[39,21,8]&{\bf[40,21,8]}&{\rm Lemma\ \ref{lem-inequality-1}}\\

[35,18,8] & {\bf[36,18,8]} & {\rm Lemma\ \ref{lem-inequality-1}}&[39,22,8]&{\bf[40,22,8]} &{\rm Lemma\ \ref{lem-inequality-1}}\\

[35,19,7] & [36,19,7] & {\rm Lemma\ \ref{lem-inequality-1}}&
[39,23,7]&[40,23,7]&{\rm Lemma\ \ref{lem-inequality-1}}\\
\hline
\end{tabular}
\end{center}}


\begin{remark}
We give Tables 4 and 5
 by combining Tables 1-3, \cite[Tables 1-3]{2-LCD-40} and \cite[Tables 1-2]{B-LCD-40}.
Furthermore, the diamond ``$\diamond$'' indicates that the construction method of the corresponding binary LCD code is different from that of \cite{2-LCD-40}.
The asterisk ``$*$'' indicates that the corresponding binary LCD code has new parameters comparing with \cite{B-LCD-40} and \cite{2-LCD-40}.
\end{remark}

{\small
\begin{center}
\begin{tabular}{cccccccccc}
\multicolumn{10}{c}{{\rm Table 4: Bounds on the minimum diatance of binary LCD codes}}\\
\hline
    \makebox[0.07\textwidth][c]{$n\backslash k$}& \makebox[0.07\textwidth][c]{9} & \makebox[0.07\textwidth][c]{10}& \makebox[0.07\textwidth][c]{11}& \makebox[0.07\textwidth][c]{12}& \makebox[0.07\textwidth][c]{13}& \makebox[0.07\textwidth][c]{14}& \makebox[0.07\textwidth][c]{15}& \makebox[0.07\textwidth][c]{16}& \makebox[0.07\textwidth][c]{17}\\
    \hline\hline
29   & 10& 10 &${\bf9^\diamond}$ & 8& 8  & 8  & 6  &  6 & 6\\

30   & 11& 10 &${\bf9^\diamond}$-10 & 9& 8  & 8  & ${\bf7^*} $ &  6 & 6\\

31   & 11& 10 &  ${\bf10^\diamond}$ &10&9&8 & ${\bf7^\diamond}$-8& ${\bf 7^*}$ & 6 \\

32   & 12& 11 &  10  & 10&9-10&8-9 & ${\bf8^*}$& ${\bf8^*}$ & ${\bf 7^*}$\\

33  & 12& 12 &  10-11 & 10&10& 9-10& ${\bf8^\diamond}$-9& ${\bf 8^\diamond} $& ${\bf 7^*}$-8\\

34   & 13& 12 &  11-12  & 11&10& 10&  ${\bf9^\diamond}$-10& ${\bf8^\diamond}$-9 & ${\bf8^*}$\\

35   & 13-14& 12-13 &  12  & 12&10-11& 10& 9-10& 9-10 & ${\bf8^*}$\\

36   & 14& 12-14 &  12-13  & 12&10-12& 10-11& ${\bf10^\diamond}$& 10 & ${\bf8^\diamond}$-9\\

37   & 14& ${\bf13^*}$-14 &  12-14& 12-13&10-12& 10-12&10-11&10 & ${\bf9^\diamond}$-10\\

38  & 14-15& ${\bf14^*}$ &  12-14& 12-14&11-12& 10-12&10-12&10-11 &${\bf9^\diamond}$-10\\

39   & 14-16& ${\bf14^\diamond}$-15 &  ${\bf13^\diamond}$-14& 12-14&11-13& 11-12&10-12&10-12 & 10-11\\

40   & 15-16& ${\bf14^\diamond}$-16 &  ${\bf13^\diamond}$-15& 13-14&12-14& 11-13&10-12&10-12 & 10-12\\
\hline
\end{tabular}
\end{center}

}

{\small
\begin{center}
\begin{tabular}{cccccccccccccc}
\multicolumn{14}{c}{{\rm Table 5: Bounds on the minimum diatance of binary LCD codes}}\\
\hline
    \makebox[0.05\textwidth][c]{$n\backslash k$}& \makebox[0.05\textwidth][c]{18} & \makebox[0.05\textwidth][c]{19}& \makebox[0.05\textwidth][c]{20}& \makebox[0.05\textwidth][c]{21}& \makebox[0.05\textwidth][c]{22}& \makebox[0.05\textwidth][c]{23}& \makebox[0.05\textwidth][c]{24}& \makebox[0.03\textwidth][c]{25}& \makebox[0.03\textwidth][c]{26}& \makebox[0.03\textwidth][c]{27}& \makebox[0.03\textwidth][c]{28}& \makebox[0.03\textwidth][c]{29}& \makebox[0.03\textwidth][c]{30}\\
    \hline\hline
29   & 6  & 5& 4  & 4  & 4  & 3&     2    & 2  & 2  & 2  & 2  &  1 & \\

30   & 6  & 5-6& 5  & 4  & 4  & 4 &     3     & 2  & 2  & 2  & 2  &  1 & 1\\

31   & 6& 6& 6& ${\bf5^\diamond}$ & 4  & 4  &  4    &   3   &2&2 & 2& 2 & 2 \\

32   & 6& 6& 6&  ${\bf5^\diamond}$-6& 5& 4 &  4     & 3-4&3&2 & 2& 2 & 2\\

33   & ${\bf 7^*}$& 6& 6&  ${\bf6^\diamond}$&  6&  5 & 4  & 4&4& 3& 2& 2 & 2\\

34   & ${\bf 8^*}$& 6-7& 6& ${\bf6^\diamond}$&  6&  5-6 &  4  & 4&4& 3-4& 3&2 & 2\\

35   & ${\bf 8^*}$& ${\bf 7^*}$-8& 6-7& ${\bf 6^*}$&  6& ${\bf 6^\diamond} $&   5  & 4&4& 4& 4& 3 & 2\\

36   &  ${\bf8^*}$& ${\bf7^\diamond}$-8& ${\bf 7^*}$-8&  ${\bf6^\diamond}$-7&  6&  ${\bf6^\diamond}$ &  6  &  5&4& 4& 4& 3-4 & 3\\

37   &  ${\bf8^\diamond}$-9& ${\bf 8^*}$& ${\bf 8^*}$& ${\bf 7^*}$-8& 6-7&  ${\bf6^\diamond}$ & 6&  5-6& 5& 4&4&4 & 4\\

38   & 9-10&  ${\bf8^\diamond}$-9& $ {\bf8^\diamond}$& $ {\bf7^\diamond}$-8& ${\bf7^*}$-8& ${\bf6^\diamond}$-7 &   6& 6& 6&  5&4&4 & 4\\

39   & 10&  9-10& ${\bf8^\diamond}$-9& ${\bf 8^*}$& ${\bf8^*}$& ${\bf7^*}$-8 &   6-7&  6&6&  5-6& 5&4 & 4\\

40   & 10-11& 9-10&  9-10& ${\bf 8^*}$-9&${\bf8^\diamond}$ & ${\bf7^\diamond}$-8 &  6-8&  6-7 & 6 &  6& 6&5 & 4\\
\hline
\end{tabular}
\end{center}

}

\section{New ternary LCD codes}

The constructions of optimal ternary LCD $[n,k]$ codes were studies
in \cite{AH-TLCD-1-10,ter-11-19,234-lcD}. Hence, we will construct some unknown ternary LCD codes and extend a result of Araya, Harada and Saito to length 25 by Constructions I and II in Section 3.

\subsection{Some important inequalities}
Let $d^E_3(n,k)$ denote the largest minimum weight among all ternary Euclidean LCD $[n, k]$ codes. By adding zero column, we have $d^E_3 (n+1, k)\geq d^E_3 (n, k)$.

\begin{lem}{\rm\cite{34-LCD}}\label{lemma-leq-ternary}
Suppose that $2\leq k\leq n$. Then
$d^E_3(n,k)\leq d^E_3(n,k-1).$
\end{lem}

\begin{lem}{\rm\cite{ter-11-19}}
If $20\leq n\leq 25$, then we have
$$d_3^E(n,n-2)=2,\ d_3^E(n,n-3)=2,\ d_3^E(n,n-4)=3.$$
\end{lem}

According to \cite{AH-TLCD-1-10}, we know that $d_3^E(n,n-1)=1$ if $n\equiv 0\ ({\rm mod}\ 3)$ and $d_3^E(n,n-1)=2$ if $n\equiv 1,2\ ({\rm mod}\ 3)$.
The following proposition is a generalization of \cite[Proposition 5]{B-LCD-40}.

\begin{prop}\label{prop-ternary}
Let $C$ be a ternary linear $[n,k,d]$ code and $\dim(C\cap C^\perp)=s$. If $C_i$ is the shortened code of $C$ in $i$-th coordinate for some $1\in \{1,\ldots,n\}$, then $dim(C_i\cap C_i^\perp)\leq s+1$.
\end{prop}

\begin{proof}
Without loss of generality, we can consider the shortened code $C_1$ of $C$ on the first coordinate. If all codewords have $0$ as a first coordinate, then $\dim(C_1\cap C_1)=\dim(C\cap C^\perp)=s$.
Otherwise $C=(0|C_1)\cup(1|u+C_1)\cup(2|2u+C)$ for a codeword $(1,u)\in C$.
Let $\mathcal{H}=C\cap C^\perp$ and $\mathcal{H}_1=C_1\cap C_1^\perp$.
There are two possibilities for $\mathcal{H}$, namely $\mathcal{H}=(0|\mathcal{H}')$ or $\mathcal{H}=(0|\mathcal{H}')\cup(1|v+\mathcal{H}')\cup(2|2v+\mathcal{H}')$.
In both cases, $\mathcal{H}'\subseteq\mathcal{H}_1$.

(i) If $\mathcal{H}'=\mathcal{H}_1$, then $\dim(\mathcal{H}_1)=\dim(\mathcal{H})$ or $\dim(\mathcal{H})-1$.

 (ii) If $\mathcal{H}'\subsetneqq\mathcal{H}_1$. Take $y_1,y_2\in \mathcal{H}_1\backslash \mathcal{H}'\subseteq C_1\cap C_1^\perp$.
      It follows from $y_i\in C_1$ that $(0,y_i)\in C$.
      If follows from $y_i\in C_1^\perp=(C^\perp)^1$ that $(\lambda_i,y_i)\in C^\perp$ for some $\lambda_i\in \F_3$.
      If $\lambda_i=0$, then $(0,y_i)\in C\cap C^\perp=\mathcal{H}$. Implying that $y_i\in \mathcal{H}'$, which is a contradiction. Hence $\lambda_i\neq 0$. So $(0,\lambda_2y_1-\lambda_1y_2)\in \mathcal{H}$, which implies that $\lambda_2y_1-\lambda_1y_2\in \mathcal{H}'$. It turns out that $\mathcal{H}_1=\mathcal{H}'\cup(y_1+\mathcal{H}')\cup(2y_1+\mathcal{H}')$ and $\dim(\mathcal{H}_1)=\dim(\mathcal{H}')+1$.
      Since $\dim(\mathcal{H}')=\dim(\mathcal{H})$ or $\dim(\mathcal{H})-1$, we have $\dim(\mathcal{H}_1)=\dim(\mathcal{H})$ or $\dim(\mathcal{H})+1$.
 This completes the proof.
\end{proof}

\begin{cor}\label{3-cor-leq}
Let $k$ and $n$ be two integers such that $2\leq k\leq n$. Then we have
$$d^E_3(n,k)\leq \max\{d^E_3(n-1,k-1),d^E_3(n-2,k-2)\}.$$
\end{cor}
\begin{proof}
The proof is similar to that of Corollary \ref{cor-4.4}. The main difference is that we use Proposition \ref{prop-ternary} instead of \cite[Proposition 5]{B-LCD-40}.
\end{proof}

\subsection{New ternary LCD codes}

Firstly, we give some known ternary LCD codes from \cite{ter-11-19}.
\begin{center}
$13\leq d_3(23,4)\leq 14$, $d_3(24,4)=15$, $15\leq d_3(25,4)\leq 16$.
\end{center}

According to \cite{ter-11-19}, there are ternary LCD codes $C_{[19,6,9]_3}$, $C_{[20,5,11]_3}$, $C_{[20,6,10]_3}$ and $C_{[20,8,8]_3}$, they have the following generator matrices respectively.
\begin{center}
{\small
$
\begin{pmatrix}
\setlength{\arraycolsep}{1.2pt}
\begin{array}{ccccccccccccccccccc}
1&0&0&0&0&0&0&0&0&0&0&1&1&1&1&1&1&1&1\\
0&1&0&0&0&0&0&0&0&1&1&2&2&2&1&1&1&0&0\\
0&0&1&0&0&0&1&1&1&2&1&1&1&1&0&0&0&0&0\\
0&0&0&1&0&0&1&2&2&2&1&2&0&0&1&1&0&0&0\\
0&0&0&0&1&0&1&2&0&1&2&1&1&0&1&0&0&1&0\\
0&0&0&0&0&1&1&1&2&1&0&1&0&1&2&0&1&1&0
\end{array}
\end{pmatrix}
$,\
$
\begin{pmatrix}
\setlength{\arraycolsep}{1.2pt}
\begin{array}{cccccccccccccccccccc}
1&0&0&0&0&0&0&0&0&0&1&1&1&1&1&1&1&1&1&1\\
0&1&0&0&0&0&0&1&1&1&0&0&0&1&1&1&1&2&2&2\\
0&0&1&0&0&0&1&2&2&1&2&2&1&0&2&1&1&2&2&1\\
0&0&0&1&0&1&1&2&1&1&0&2&2&2&2&1&0&0&0&0\\
0&0&0&0&1&1&2&0&1&2&2&0&2&1&2&1&2&1&0&1
\end{array}
\end{pmatrix}
$,}
\end{center}
\begin{center}
{\small$
\begin{pmatrix}
\setlength{\arraycolsep}{1.2pt}
\begin{array}{cccccccccccccccccccc}
1&0&0&0&0&0&0&0&0&0&0&1&1&1&1&1&1&1&1&1\\
0&1&0&0&0&0&0&0&1&1&1&2&2&1&1&1&1&0&0&0\\
0&0&1&0&0&0&1&1&2&1&1&1&1&1&0&0&0&1&0&0\\
0&0&0&1&0&0&1&2&0&2&1&1&1&0&1&1&0&0&1&0\\
0&0&0&0&1&0&0&1&1&1&2&1&0&0&2&0&1&1&1&0\\
0&0&0&0&0&1&1&2&2&0&0&1&2&2&1&2&0&0&0&1
\end{array}
\end{pmatrix},\
\begin{pmatrix}
\setlength{\arraycolsep}{1.2pt}
\begin{array}{cccccccccccccccccccc}
1&0&0&0&0&0&0&0&0&0&0&0&0&1&1&1&1&1&1&1\\
0&1&0&0&0&0&0&0&0&0&1&1&1&1&1&1&1&0&0&0\\
0&0&1&0&0&0&0&0&1&1&1&2&2&1&1&0&0&0&0&0\\
0&0&0&1&0&0&0&0&1&2&2&1&0&1&0&2&1&0&0&0\\
0&0&0&0&1&0&0&0&1&2&2&2&1&0&0&1&0&1&0&0\\
0&0&0&0&0&1&0&0&1&1&2&0&2&2&0&2&0&0&2&0\\
0&0&0&0&0&0&1&0&1&2&0&1&2&2&0&0&0&1&1&0\\
0&0&0&0&0&0&0&1&0&2&0&2&0&2&2&2&1&1&1&0
\end{array}
\end{pmatrix}.
$
}
\end{center}

\begin{prop}
There are ternary LCD $[20,7,9]$ and $[20,12,6]$ codes.
\end{prop}

\begin{proof}
We start from the ternary LCD code $C_{[19,6,9]_3}$. By applying Theorem \ref{Methods}, we can construct a ternary LCD $[20,7,9]$ code $C_{[20,7,9]_3}$ with the generator matrix
{\small$$
G_{[20,7,9]_3}=\begin{pmatrix}
\begin{array}{c|ccccccccccccccccccc}
1&1&1&0&2&0&0&1&1&0&0&0&0&0&1&1&0&2&2&2\\
\hline
0&1&0&0&0&0&0&0&0&0&0&0&1&1&1&1&1&1&1&1\\
0&0&1&0&0&0&0&0&0&0&1&1&2&2&2&1&1&1&0&0\\
0&0&0&1&0&0&0&1&1&1&2&1&1&1&1&0&0&0&0&0\\
0&0&0&0&1&0&0&1&2&2&2&1&2&0&0&1&1&0&0&0\\
0&0&0&0&0&1&0&1&2&0&1&2&1&1&0&1&0&0&1&0\\
0&0&0&0&0&0&1&1&1&2&1&0&1&0&1&2&0&1&1&0
\end{array}
\end{pmatrix}.
$$}
By the MAGMA function BKLC, one can construct a ternary LCD $[20,12,6]$ code.
\end{proof}

\begin{remark}
We list in Tables 6-7 some ternary LCD codes obtained by Theorems \ref{LCD-(n-l,k-l)}, \ref{puncture}, \ref{Methods} and \ref{Methods-2}. To save the space, the codes in Table 6 can be obtained from {\tt https://ahu-coding.github.io/code1/}.
\end{remark}

{\small
\begin{center}
 \begin{threeparttable}
\begin{tabular}{llll}
\multicolumn{4}{c}{{\rm Table 6: The ternary LCD codes from Theorems \ref{LCD-(n-l,k-l)} and \ref{puncture}}}\\
\hline
    \makebox[0.25\textwidth][l]{{\rm The\ code}\ $C_{[n,k,d]_3}^{\ell}$}& \makebox[0.25\textwidth][l]{The set $T$} &
    \makebox[0.2\textwidth][l]{$C_T$} &
    \makebox[0.15\textwidth][l]{References}\\
    \hline\hline
$C^{0}_{[20,12,6]_3}$ &  $\backslash$& [20,12,6]& Theorem \ref{LCD-(n-l,k-l)}\\

$C^{0}_{[21,12,6]_3} $& $\backslash$ & [21,12,6]& Theorem \ref{LCD-(n-l,k-l)}\\

$C^{0}_{[21,15,4]_3}$ & $\backslash$ & [21,15,4]& Theorem \ref{LCD-(n-l,k-l)}\\

$C^{0}_{[21,17,3]_3}$ & $\backslash$ & [21,17,3]& Theorem \ref{LCD-(n-l,k-l)}\\

$C^{3}_{[25,9,11]_3} $& \{1,2,3\} & [22,6,11]& Theorem \ref{LCD-(n-l,k-l)}\\

$C^{3}_{[25,13,8]_3}$ & \{1,2,9\} & [22,10,8]& Theorem \ref{LCD-(n-l,k-l)}\\

$C^{0}_{[23,17,4]_3}$ &$ \backslash$& [23,17,4]& Theorem \ref{LCD-(n-l,k-l)}\\

$C^{0}_{[24,18,4]_3} $& $\backslash$ & [24,18,4]& Theorem \ref{LCD-(n-l,k-l)}\\

$C^{9}_{[44,29,8]_3} $& \{2,3,4,5,6,7,8,9,10\} & [35,20,8]& Theorem \ref{LCD-(n-l,k-l)}\\

$C^{1}_{[40,29,6]_3}$ & \{2\} & [39,28,6]& Theorem \ref{LCD-(n-l,k-l)}\\
\hline\hline
{\rm The\ code}\ $C_{[n,k,d]_3}^{\ell}$& The set $T$ &
    $C^T$ & References\\
    \hline\hline
$C^3_{[27,5,16]_3}$ & \{1,2,5\} & [24,5,13]& Theorem \ref{puncture}\\

$C^4_{[29,6,16]_3}$ & \{1,2,3,4\} & [25,6,13]& Theorem \ref{puncture}\\
\hline
\end{tabular}
\begin{tablenotes}
\footnotesize
\item Note that the code $C_{[n,k,d]_3}^{\ell}$ is the best-known ternary $[n,k,d]$ code from MAGMA \cite{magma}.
\end{tablenotes}
\end{threeparttable}
\end{center}
}

\begin{prop}
There are ternary LCD $[22,11,7]$ and $[24,16,5]$ codes.
\end{prop}

\begin{proof}
We start from the ternary LCD codes $C_{[22,10,8]_3}$ and $C_{[24,15,6]_3}$ (see Tables 6 and 7). By applying Theorem \ref{Methods-2}, we can construct ternary LCD codes $C_{[22,11,7]_3}$ and $C_{[24,16,5]_3}$ with the following generator matrices, respectively,
{\small$$
G_{[22,11,7]_3}=\begin{pmatrix}
\begin{array}{cccccccccccccccccccccc}
1101222110122211012200\\
\hline
1000010000002002210122\\
0100020000102002111010\\
0010010000101020000122\\
0001020000222012211022\\
0000110000110010202010\\
0000001000011020021021\\
0000000100221002122200\\
0000000010212022221211\\
0000000001101002102021\\
0000000000000111111110
\end{array}
\end{pmatrix},~
G_{[24,16,5]_3}=\begin{pmatrix}
\begin{array}{cccccccccccccccccccccccc}
001121202101001122102021\\
\hline
100000000000000120201111\\
010000000000000112011000\\
001000000000000112110022\\
000100000000000020110211\\
000010000000000020021220\\
000001000000000002002122\\
000000100000000012220010\\
000000010000000001222001\\
000000001000000021102102\\
000000000100000011100011\\
000000000010000022120200\\
000000000001000002212020\\
000000000000100000221202\\
000000000000010012012221\\
000000000000001022211121
\end{array}
\end{pmatrix}.
$$}
\end{proof}

{\small
\begin{center}
 \begin{threeparttable}
\begin{tabular}{llll}
\multicolumn{4}{c}{{\rm Table 7: The ternary LCD codes from Theorem \ref{Methods} and Theorem \ref{Methods-2}}}\\
\hline
    \makebox[0.18\textwidth][l]{{\rm The\ given}\ $C_{[n,k,d]_3}$}& \makebox[0.22\textwidth][l]{The vector ${\bf x}$ or ${\bf y}$} &
    \makebox[0.08\textwidth][l]{$C'$} &
    \makebox[0.1\textwidth][l]{References}\\
    \hline\hline
[19,6,9] &  (1102001100000110222) & [20,7,9] & {\rm Theorem}\ \ref{Methods} \\

[20,5,11] &  (21112201000000010021) & [21,6,10] & {\rm Theorem}\ \ref{Methods} \\

[20,6,10] & (02112000000000121221) & [21,7,9] & {\rm Theorem}\ \ref{Methods} \\

[20,8,8] &  (12021210000020212222) & [21,9,8] & {\rm Theorem}\ \ref{Methods} \\

[21,6,10] & (200002221000002020221) & [22,7,10] & {\rm Theorem}\ \ref{Methods} \\

[21,7,9] &  (222210000000002002222) & [22,8,9] & {\rm Theorem}\ \ref{Methods} \\

[21,12,6] &  (120000022222120122212) & [22,13,6] & {\rm Theorem}\ \ref{Methods} \\

[21,15,4] &  (010000111012101011211) & [22,16,4] & {\rm Theorem}\ \ref{Methods} \\

[21,17,3] &   (010021201100021201212) & [22,18,3] & {\rm Theorem}\ \ref{Methods} \\

[22,6,11] &   (1211021222210000112200) & [23,7,11] & {\rm Theorem}\ \ref{Methods} \\

[22,7,10] &   (1112202121010001202201) & [23,8,10] & {\rm Theorem}\ \ref{Methods} \\

[22,8,9] &   (2210022110000011120021) & [23,9,9] & {\rm Theorem}\ \ref{Methods} \\

[22,11,7] &   (0010002100121001211101) & [23,12,7] & {\rm Theorem}\ \ref{Methods} \\

[22,13,6] &    (2110210002112020111112) & [23,14,6] & {\rm Theorem}\ \ref{Methods} \\

[23,7,11] &    (10112100000000000102211) & [24,8,10] & {\rm Theorem}\ \ref{Methods} \\

[23,12,7] &     (00100211121000112121012) & [24,13,7] & {\rm Theorem}\ \ref{Methods} \\

[23,14,6] &     (22100010002211222202121) & [24,15,6] & {\rm Theorem}\ \ref{Methods} \\

[24,8,10] &     (020212110000000020010011) & [25,9,10] & {\rm Theorem}\ \ref{Methods} \\

[24,15,6] &    (102200010201120202000120) & [25,16,6] & {\rm Theorem}\ \ref{Methods} \\

[24,16,5] &    (010000010021012212210010) & [25,17,5] & {\rm Theorem}\ \ref{Methods} \\

[24,18,4] &    (121000002021110021222212) & [25,19,4] & {\rm Theorem}\ \ref{Methods} \\

[35,20,8] &       (12202110000000002101221001212020220) & [36,21,8] & {\rm Theorem}\ \ref{Methods} \\

[36,21,8] &        (100200212100000112200221222000012121) & [37,22,8] & {\rm Theorem}\ \ref{Methods} \\

[39,28,6] &         (021020210000200202112021120012101111112) & [40,29,6] & {\rm Theorem}\ \ref{Methods} \\
\hline
[22,10,8] &     (1101222110122211012200) & [22,11,7] & {\rm Theorem}\ \ref{Methods-2} \\

[22,13,6] &      (0020100000102202121012) & [22,14,5] & {\rm Theorem}\ \ref{Methods-2} \\

[23,14,6] &  (00020100000102202121012) & [23,15,5] & {\rm Theorem}\ \ref{Methods-2} \\

[24,15,6] &  (001121202101001122102021) & [24,16,5] & {\rm Theorem}\ \ref{Methods-2} \\

[25,9,10] &  (1112101001021000011020001) & [25,10,9] & {\rm Theorem}\ \ref{Methods-2}\\

[25,10,9] &  (1222121000002000212020021) & [25,11,8] & {\rm Theorem}\ \ref{Methods-2}\\
\hline
\end{tabular}
\begin{tablenotes}
\footnotesize
\item Note that the given code $C_{[n,k,d]_3}$ in Table 7 is from Tables 6-7.
\end{tablenotes}
\end{threeparttable}
\end{center}}

\begin{cor}
$d_3^E(21,9)=8,d_3^E(22,10)=8,d_3^E(23,11)\in \{7,8\},d_3^E(24,12)\in \{7,8\}.$
\end{cor}

\begin{proof}
By Tables 6 and 7, there are ternary LCD $[21,9,8]$, $[22,10,8]$, $[23,11,7]$ and $[24,12,7]$ codes. From \cite[Table 7]{ter-11-19}, we know that $d_3^E(19,7)=d_3^E(20,8)=8$.
It follows from Corollary \ref{3-cor-leq} that $d_3^E(21,9)\leq \max\{d_3^E(20,8),d_3^E(19,7)\}=8$.
Similarly, we have $d_3^E(22,10)\leq 8$, $d_3^E(23,11)\leq 8$ and $d_3^E(24,12)\leq 8$.
\end{proof}

\begin{remark}
Combining the Database \cite{codetables} and the above results, we give lower and upper bounds on minimal distance of ternary LCD codes with length $20\leq n\leq 25$, where the parameters for ternary LCD codes of the length 20 can be found in \cite{ter-11-19}. All results are listed in Tables 8-9.
\end{remark}

{\small
\begin{center}
 \begin{threeparttable}
\begin{tabular}{cccccccccccc}
\multicolumn{12}{c}{{\rm Table 8: Bound on the minimum distance of ternary LCD codes}}\\
\hline
    \makebox[0.05\textwidth][c]{$n\backslash k$}& \makebox[0.06\textwidth][c]{4} & \makebox[0.06\textwidth][c]{5}& \makebox[0.06\textwidth][c]{6}& \makebox[0.06\textwidth][c]{7}& \makebox[0.06\textwidth][c]{8}& \makebox[0.06\textwidth][c]{9}& \makebox[0.06\textwidth][c]{10}& \makebox[0.05\textwidth][c]{11}& \makebox[0.05\textwidth][c]{12}& \makebox[0.05\textwidth][c]{13}& \makebox[0.05\textwidth][c]{14}\\
    \hline\hline
20   & 12 & 11& 10& {\bf 9}& 8& 7-8& 7 &6 & {\bf6}& 5  & 4     \\

21   & 12 & 11-12& 10-11& 9-10& 8-9& 8& 7-8 &6-7 & 6& 5-6  & 4-5 \\

22   & 12-13 & 11-12& 11-12& 10-11& 9-10& 8-9& 8 & 7-8 & 6-7& 6  & 5-6   \\

23   & 13-14 &11-13& 11-12&  11-12& 10-11& 9-10& 8-9 &7-8 & 7-8& 6-7  & 6     \\

24   & 15 & 13-14& 11-13&  11-12& 10-11& 9-11& 8-10 &7-9 & 7-8& 7-8  & 6-7    \\

25   & 15-16 &  13-15& 13-14&  11-13&  10-12& 10-11& 9-11 &8-10 & 7-9& 7-8  & 6-8    \\
\hline
\end{tabular}
\begin{tablenotes}
\footnotesize
\item The parameters in bold denote the corresponding code has new parameters according to \cite{ter-11-19}.
\end{tablenotes}
\end{threeparttable}
\end{center}

}

{\small
\begin{center}
\begin{tabular}{cccccccccccc}
\multicolumn{12}{c}{{\rm Table 9: Bound on the minimum distance of ternary LCD codes}}\\
\hline
    \makebox[0.05\textwidth][c]{$n\backslash k$}& \makebox[0.06\textwidth][c]{15} & \makebox[0.06\textwidth][c]{16}& \makebox[0.06\textwidth][c]{17}& \makebox[0.06\textwidth][c]{18}& \makebox[0.06\textwidth][c]{19}& \makebox[0.06\textwidth][c]{20}& \makebox[0.06\textwidth][c]{21}& \makebox[0.05\textwidth][c]{22}& \makebox[0.05\textwidth][c]{23}& \makebox[0.05\textwidth][c]{24}& \makebox[0.05\textwidth][c]{25}\\
    \hline\hline
20   & 3-4& 3& 2& 2 & 2& 1 &  &  &  &   &       \\

21   & 4& 3& 3& 2 & 2& 1 & 1 &  &  &   &      \\

22   & 4-5& 4& 3& 3 & 2& 2& 2 & 1 &  &   &      \\

23   & 5-6& 4-5& 4& 3 & 3& 2& 2& 2 & 1 &   &        \\

24   & 6& 5-6& 4-5& 4 & 3& 3& 2& 2& 1 &  1 &    \\

25  & 6-7& 6 & 5-6& 4-5 & 4& 3& 3& 2& 2& 2  & 1 \\
\hline
\end{tabular}
\end{center}
}


\begin{prop}
There are ternary LCD $[37,22,8]$ and $[40,29,6]$ codes.
\end{prop}

\begin{proof}
We start from the ternary LCD code $C_{[35,20,8]_3}$ (see Table 6). By applying Theorem \ref{Methods}, we can construct a ternary LCD $[36,21,8]$ code $C_{[36,21,8]_3}$, where ${\bf x}=(122021100000000\\02101221001212020220)$.
We start from the ternary LCD code $C_{[36,21,8]_3}$. By applying Theorem \ref{Methods}, we can construct a ternary LCD $[37,22,8]$ code $C_{[37,22,8]_3}$ with the generator matrix
{\small$$
G_{[37,22,8]_3}=\begin{pmatrix}
\begin{array}{c|c|c}
1&1&00200212100000112200221222000012121\\
\hline
0&1&12202110000000002101221001212020220\\
\hline
0&0&1 0 0 0 0 0 0 0 0 0 0 0 0 0 0 0 0 0 0 2 0 2 2 2 2 1 0 1 0 1 0 1 1 1 0\\
0&0&0 1 0 0 0 0 0 0 0 0 0 0 0 0 0 0 0 0 0 1 0 0 0 1 0 2 1 2 1 1 1 0 2 2 2\\
0&0&0 0 1 0 0 0 0 0 0 0 0 0 0 0 0 0 0 0 0 1 0 1 2 0 1 1 2 1 2 0 1 2 2 2 0\\
0&0&0 0 0 1 0 0 0 0 0 0 0 0 0 0 0 0 0 0 0 1 0 2 0 1 0 1 2 1 1 0 0 2 2 2 0\\
0&0&0 0 0 0 1 0 0 0 0 0 0 0 0 0 0 0 0 0 0 1 0 1 1 2 0 2 2 1 2 1 1 0 2 2 0\\
0&0&0 0 0 0 0 1 0 0 0 0 0 0 0 0 0 0 0 0 0 0 0 1 1 0 1 1 0 1 1 2 2 1 2 1 1\\
0&0&0 0 0 0 0 0 1 0 0 0 0 0 0 0 0 0 0 0 0 1 0 0 0 2 2 2 0 2 1 1 2 0 1 2 1\\
0&0&0 0 0 0 0 0 0 1 0 0 0 0 0 0 0 0 0 0 0 2 0 1 0 1 2 2 1 1 0 0 0 1 2 1 1\\
0&0&0 0 0 0 0 0 0 0 1 0 0 0 0 0 0 0 0 0 0 1 0 0 1 0 0 2 2 1 1 1 2 0 1 2 2\\
0&0&0 0 0 0 0 0 0 0 0 1 0 0 0 0 0 0 0 0 0 1 0 0 0 0 0 0 2 2 0 1 2 0 2 0 2\\
0&0&0 0 0 0 0 0 0 0 0 0 1 0 0 0 0 0 0 0 0 1 0 2 2 1 1 0 1 2 1 2 2 1 0 2 0\\
0&0&0 0 0 0 0 0 0 0 0 0 0 1 0 0 0 0 0 0 0 2 0 0 2 2 1 0 1 0 1 1 1 0 2 2 0\\
0&0&0 0 0 0 0 0 0 0 0 0 0 0 1 0 0 0 0 0 0 0 0 1 2 2 1 0 0 2 0 0 0 2 2 1 0\\
0&0&0 0 0 0 0 0 0 0 0 0 0 0 0 1 0 0 0 0 0 1 0 1 0 1 0 2 1 0 0 0 1 0 1 2 1\\
0&0&0 0 0 0 0 0 0 0 0 0 0 0 0 0 1 0 0 0 0 2 0 1 1 1 1 0 1 2 1 2 2 0 2 1 1\\
0&0&0 0 0 0 0 0 0 0 0 0 0 0 0 0 0 1 0 0 0 2 0 0 1 2 1 1 0 1 0 1 1 1 1 0 1\\
0&0&0 0 0 0 0 0 0 0 0 0 0 0 0 0 0 0 1 0 0 1 0 1 1 2 0 1 0 0 1 2 0 0 0 0 1\\
0&0&0 0 0 0 0 0 0 0 0 0 0 0 0 0 0 0 0 1 0 2 0 2 0 0 2 1 1 1 1 1 2 0 0 2 0\\
0&0&0 0 0 0 0 0 0 0 0 0 0 0 0 0 0 0 0 0 1 0 0 0 1 2 1 2 1 2 1 2 1 0 2 0 1\\
0&0&0 0 0 0 0 0 0 0 0 0 0 0 0 0 0 0 0 0 0 0 1 0 1 2 0 2 1 1 0 1 2 2 0 0 1
\end{array}
\end{pmatrix}.
$$}
We start from the ternary LCD code $C_{[39,28,6]_3}$ (see Table 6). By applying Theorem \ref{Methods}, we can construct a ternary LCD $[40,29,6]$ code $C_{[40,29,6]_3}$, where ${\bf x}=(021020210000200202112\\021120012101111112)$.
\end{proof}

\begin{remark}
We can construct some ternary LCD codes with better parameters compared with \cite{234-lcD}.
For example, the ternary LCD code of the length 37 with the dimension 22 has the minimum distance 8, while the ternary LCD code of the length 37 with the dimension 22 in \cite[Corollary 6.3]{234-lcD} has the minimum distance 7. That is to say, the ternary LCD code $C_{[37,22,8]_3}$ we obtained is also considered new.
The code $C_{[40,29,6]}$ also has better parameters compared with \cite[Corollary 6.3]{234-lcD}.
\end{remark}

\section{New quaternary Hermitian LCD codes}
Let $\F_4=\{0,1,\omega,\omega^2\}$.
Let $d^H_4 (n, k)$ denote the largest minimum distance among all quaternary Hermitian LCD $[n, k]$ codes.

\begin{prop}\label{H-22,23,24}
There are quaternary Hermitian LCD $[22,12,7]$, $[23,13,7]$, $[24,14,7]$ and $[25,15,7]$ codes.
\end{prop}

\begin{proof}
By the MAGMA function BKLC, one can construct a quaternary $[25,15,7]$ code $C^{4}_{[25,15,7]_4}$ with 4-dimensional Hermitian hull. According to Theorem \ref{LCD-(n-l,k-l)},
the shortened code $(C^{4}_{[25,15,7]_4})_{T}$ on $T=\{1,2,3,4\}$ is a quaternary Hermitian LCD $[21,11,7]$ code $C_{[21,11,7]_4}$, which has the following generator matrix
$$G_{[21,11,7]_4}=\left(
\setlength{\arraycolsep}{1.5pt}
                    \begin{array}{cccccccccccccccccccccccccccccc}
1&0&0&0&0&0&0&0&0&0&0&\omega&\omega^2&0&\omega^2&\omega^2&\omega&1&\omega&1&\omega^2\\
0&1&0&0&0&0&0&0&0&0&0&1&0&1&1&1&\omega^2&0&\omega&\omega&\omega\\
0&0&1&0&0&0&0&0&0&0&0&\omega^2&1&\omega^2&\omega^2&0&0&\omega^2&1&1&\omega^2\\
0&0&0&1&0&0&0&0&0&0&0&\omega^2&1&\omega&\omega^2&\omega^2&\omega&0&\omega&\omega^2&0\\
0&0&0&0&1&0&0&0&0&0&0&\omega^2&\omega^2&\omega&0&\omega&1&0&\omega&\omega^2&\omega\\
0&0&0&0&0&1&0&0&0&0&0&0&\omega^2&1&0&\omega&\omega&\omega^2&0&\omega&1\\
0&0&0&0&0&0&1&0&0&0&0&\omega&\omega^2&1&1&0&\omega&0&0&0&\omega\\
0&0&0&0&0&0&0&1&0&0&0&\omega&\omega^2&\omega&0&\omega^2&\omega^2&\omega^2&\omega^2&\omega^2&\omega^2\\
0&0&0&0&0&0&0&0&1&0&0&0&\omega&0&1&0&\omega^2&\omega&\omega&0&\omega^2\\
0&0&0&0&0&0&0&0&0&1&0&1&0&\omega^2&\omega^2&\omega&0&\omega&1&\omega^2&\omega\\
0&0&0&0&0&0&0&0&0&0&1&1&1&\omega&\omega&1&0&0&\omega&1&\omega
                    \end{array}
                  \right).
$$
By applying Theorem \ref{Methods}, one can construct a quaternary Hermitian LCD $[22,12,7]$ code, where ${\bf x}=(1 \omega 1 1 0 0 0 0 0 \omega 0 \omega \omega^2 \omega \omega 0 0 \omega^2 \omega^2 \omega^2 0)$.

By the MAGMA function BKLC, one can construct a quaternary $[26,16,7]$ code $C^{3}_{[26,16,7]_4}$ with 3-dimensional Hermitian hull. According to Theorem \ref{LCD-(n-l,k-l)},
the shortened code $(C^{3}_{[26,16,7]_4})_{T}$ on $T=\{1,2,3\}$ is a quaternary Hermitian LCD $[23,13,7]$ code $C_{[23,13,7]_4}$, which has the following generator matrix
$$G_{[23,13,7]_4}=\left(
\setlength{\arraycolsep}{1.5pt}
                    \begin{array}{cccccccccccccccccccccccccccccc}
1&0&0&0&0&0&0&0&0&0&0&0&0&0&\omega^2&\omega^2&0&\omega^2&\omega^2&0&1&0&1\\
0&1&0&0&0&0&0&0&0&0&0&0&0&\omega&0&1&\omega^2&0&\omega^2&0&0&1&1\\
0&0&1&0&0&0&0&0&0&0&0&0&0&\omega^2&1&0&1&\omega&1&\omega^2&1&\omega^2&\omega^2\\
0&0&0&1&0&0&0&0&0&0&0&0&0&\omega^2&\omega^2&0&0&\omega^2&0&\omega&\omega^2&0&\omega\\
0&0&0&0&1&0&0&0&0&0&0&0&0&\omega^2&\omega&0&\omega^2&1&\omega^2&1&1&\omega^2&\omega\\
0&0&0&0&0&1&0&0&0&0&0&0&0&1&\omega&\omega^2&\omega&\omega&\omega^2&1&1&1&\omega\\
0&0&0&0&0&0&1&0&0&0&0&0&0&\omega^2&1&0&\omega&\omega&\omega^2&0&\omega&1&0\\
0&0&0&0&0&0&0&1&0&0&0&0&0&\omega&1&\omega^2&0&0&\omega&\omega&1&0&1\\
0&0&0&0&0&0&0&0&1&0&0&0&0&\omega&\omega&\omega&\omega^2&1&1&1&\omega&1&1\\
0&0&0&0&0&0&0&0&0&1&0&0&0&\omega&0&1&0&\omega^2&\omega&\omega&0&\omega^2&0\\
0&0&0&0&0&0&0&0&0&0&1&0&0&\omega^2&\omega^2&\omega&\omega&1&\omega^2&0&0&\omega^2&\omega^2\\
0&0&0&0&0&0&0&0&0&0&0&1&0&\omega&\omega&\omega^2&1&1&1&\omega^2&\omega&\omega^2&\omega^2\\
0&0&0&0&0&0&0&0&0&0&0&0&1&\omega^2&0&1&0&1&1&1&\omega^2&1&\omega^2
                    \end{array}
                  \right).
$$

By applying Theorem \ref{Methods} to the code $C_{[23,13,7]_4}$, one can construct a quaternary Hermitian LCD $[24,14,7]$ code $C_{[24,14,7]_4}$, where ${\bf x}=(0 \omega \omega^2 0 \omega 0 0 1 \omega 0 1 \omega \omega \omega^2 \omega^2 \omega^2 0 0 \omega 1 1 \omega \omega)$.

By applying Theorem \ref{Methods} to the code $C_{[24,14,7]_4}$, one can construct a quaternary Hermitian LCD $[25,15,7]$ code $C_{[25,15,7]_4}$, where ${\bf x}=(1 1 0 0 1 \omega^2 1 \omega \omega^2 0 \omega^2 \omega^2 1 0 0 \omega^2 \omega^2 \omega \omega \omega^2 0 1 \omega^2 \omega)$.
\end{proof}

\begin{remark}
Compared with Table 3 in \cite{H-25}, the quaternary Hermitian LCD codes in Lemma \ref{H-22,23,24} have better parameters than their parameters.
For example,
the quaternary Hermitian LCD code of the length 24 with the dimension 14 has the minimum distance 7, while the quaternary Hermitian LCD code of the length 24 with the dimension 14 in \cite[Table 3]{H-25} has the minimum distance 6. That is to say, the quaternary Hermitian LCD code $C_{[24,14,7]_4}$ we obtained is also considered new.
\end{remark}

As an application, we can construct some EAQECCs with different parameters from quaternary Hermitian LCD codes.
We use $[[n,k,d;c]]_2$ to denote a binary entanglement-assisted quantum error correcting code (EAQECC) that encodes $k$ information qubits into $n$ channel qubits with the help of $c$ pre-shared entanglement pairs, and $d$ is called the minimum distance of the EAQECC. EAQECCs were introduced by Brun et al. in \cite{EAQECC}, which include the standard quantum stabilizer codes as a special case.
It has shown that if there is a quaternary Hermitian
LCD $[n, k, d]$ code, then there is a binary EAQECC with parameters $[[n, k, d; n-k]]_2$ (see \cite[Lemma 2.1']{H-20}). Hence, as a consequence of Proposition \ref{H-22,23,24} and \cite[Lemma 3.3]{H-lCD-k=3}, we have that

\begin{cor}\label{cor-EAQECC}
For $(n,k,d)\in \{(22,12,7),(23,13,7),(24,14,7),(25,15,7)\}$ and a nonnegative integer $s$, there is a binary EAQECC with parameters $[[n+\frac{4^k-1}{3}s,k,d+4^{k-1}s;n+\frac{4^k-1}{3}s-k]]_2$.
\end{cor}

\begin{proof}
Combining with Proposition \ref{H-22,23,24} and \cite[Lemma 3.3]{H-lCD-k=3}, we know that there is a quaternary Hermitian LCD $[n+\frac{4^k-1}{3}s,k,d+4^{k-1}s]$ code for $(n,k,d)\in \{(22,12,7),(23,13,7),\\(24,14,7),(25,15,7)\}$ and a nonnegative integer $s$.
According to \cite[Lemma 2.1']{H-20}, there is a binary EAQECC with parameters $[[n+\frac{4^k-1}{3}s,k,d+4^{k-1}s;n+\frac{4^k-1}{3}s-k]]_2$.
\end{proof}

\begin{example}
According to Corollary \ref{cor-EAQECC}, there are binary EAQECCs with parameters $[[22,12,7;10]]_2$, $[[23,13,7;10]]_2$, $[[24,14,7;10]]_2$ and $[[25,15,7;10]]_2$. However, the best known binary EAQECCs in \cite{codetables} have parameters $[[22,12,6;6]]_2$, $[[23,13,6;7]]_2$, $[[24,14,6;7]]_2$ and $[[25,15,6;8]]_2$. The EAQECCs we construct have larger minimum distances.
\end{example}

\section{Conclusion}

In this paper, we have introduced some methods for constructing LCD codes over small finite fields by modifying some typical methods.
We have constructed many new binary LCD codes, ternary LCD codes and quaternary Hermitian LCD codes, which improve the known lower bounds on the largest minimum weights.
As a consequence, we used two counterexamples to disprove the conjecture proposed by Bouyuklieva.
Finally, as an application of quaternary Hermitian LCD codes, we found some binary EAQECCs with new parameters. We believe that our methods can produce more results for LCD codes.

\section*{Conflict of Interest}
The authors have no conflicts of interest to declare that are relevant to the content of this
article.

\section*{Data Deposition Information}
Our data can be obtained from the authors upon reasonable request.

\section*{Acknowledgement}
The authors would like to thank Dr. Hongwei Zhu for helpful discussions.
This research is supported by Natural Science Foundation of China (12071001).

\end{document}